\documentclass[twocolumn,journal,pagenumbers]{IEEEtran}
\usepackage{listings}
\usepackage{amsmath}
\usepackage{amsthm}
\usepackage{tikz}
\usepackage{caption}
\usepackage{array}
\usepackage{mdwmath}
\usepackage{multirow}
\usepackage{mdwtab}
\usepackage{eqparbox}
\usepackage{amsfonts}
\usepackage{tikz}
\usepackage{amsthm}
\usepackage{array}
\usepackage{bbm}
\usepackage{mdwmath}
\usepackage{mdwtab}
\usepackage{eqparbox}
\usepackage{tikz}
\usepackage{latexsym}
\usepackage{cite}
\usepackage{amssymb}
\usepackage{bm}
\usepackage{amssymb}
\usepackage{graphicx}
\usepackage{mathrsfs}
\usepackage{epsfig}
\usepackage{psfrag}
\usepackage{setspace}
\usepackage{hyperref}
\usepackage{algorithm}
\usepackage{algpseudocode}
\usepackage{stfloats}
\def \bP {\mathbb{P}}
\def \bE {\mathbb{E}}
\newtheorem{theorem}{Theorem}

\newtheorem{lemma}{Lemma}
\newtheorem{corollary}{Corollary}
\newtheorem{remark}{Remark}
\begin{document}
\title{Minimax Estimation of Discrete Distributions under $\ell_1$ Loss}
\author{Yanjun~Han,~\IEEEmembership{Student Member,~IEEE},~Jiantao~Jiao,~\IEEEmembership{Student Member,~IEEE}, and Tsachy~Weissman,~\IEEEmembership{Fellow,~IEEE}
\thanks{Manuscript received Month 00, 0000; revised Month 00, 0000; accepted Month 00, 0000. Date of current version Month 00, 0000.
This work was supported in part by the Center for Science of Information (CSoI), an NSF Science and Technology Center, under grant agreement CCF-0939370. The material in this paper was presented in part at the 2014 IEEE International Symposium on
Information Theory, Honolulu, HI, USA. Copyright (c) 2014 IEEE. Personal use of this material is permitted. However, permission to use this material for any other purposes must be obtained from the IEEE by sending a request to pubs-permissions@ieee.org.
}
\thanks{Yanjun Han, Jiantao Jiao and Tsachy Weissman are with the Department of Electrical Engineering, Stanford University, CA, USA. Email: \{yjhan, jiantao, tsachy\}@stanford.edu}. 
\thanks{This work was supported in part by the NSF Center for Science of Information under grant agreement CCF-0939370.}
}
\maketitle

\begin{abstract}
We consider the problem of discrete distribution estimation under $\ell_1$ loss. We provide tight upper and lower bounds on the maximum risk of the empirical distribution (the maximum likelihood estimator), and the minimax risk in regimes where the support size $S$ may grow with the number of observations $n$. We show that among distributions with bounded entropy $H$, the asymptotic maximum risk for the empirical distribution is $2H/\ln n$, while the asymptotic minimax risk is $H/\ln n$. Moreover, we show that a hard-thresholding estimator oblivious to the unknown upper bound $H$, is essentially minimax. However, if we constrain the estimates to lie in the simplex of probability distributions, then the asymptotic minimax risk is again $2H/\ln n$. We draw connections between our work and the literature on density estimation, entropy estimation, total variation distance ($\ell_1$ divergence) estimation, joint distribution estimation in stochastic processes, normal mean estimation, and adaptive estimation.
\end{abstract}


\section{Introduction and Main Results}
Given $n$ independent samples from an unknown discrete probability distribution $P=(p_1,p_2,\cdots,p_S)$, with \emph{unknown} support size $S$, we would like to estimate the distribution $P$ under $\ell_1$ loss. Equivalently, the problem is to estimate $P$ based on the Multinomal random vector $(X_1,X_2,\ldots,X_S)\sim \mathsf{Multi}(n;p_1,p_2,\ldots,p_S)$.

A natural estimator of $P$ is the Maximum Likelihood Estimator (MLE), which in this problem setting coincides with the empirical distribution $P_n$, where $P_n(i) = X_i/n$ is the number of occurrences of symbol $i$ in the sample divided by the sample size $n$. This paper is devoted to analyzing the performances of the MLE, and the minimax estimators in various regimes. Specifically, we focus on the following three regimes:

\begin{enumerate}
  \item Classical asymptotics: the dimension $S$ of the unknown parameter $P$ remains fixed, while the number of observations $n$ grows.
  \item High dimensional asymptotics: we let the support size $S$ and the number of observations $n$ grow together, characterize the scaling under which consistent estimation is feasible, and obtain the minimax rates.
  \item Infinite dimensional asymptotics: the distribution $P$ may have infinite support size, but is constrained to have bounded entropy $H(P)\le H$, where the entropy \cite{Shannon1948} is defined as
    \begin{align}
      H(P) \triangleq \sum_{i=1}^S -p_i\ln p_i.
    \end{align}
\end{enumerate}

We remark that results for the first regime follow from the well-developed theory of asymptotic statistics~\cite[Chap. 8]{Vandervaart2000}, and we include them here for completeness and comparison with other regimes. One motivation for considering the high dimensional and infinite dimensional asymptotics is that the modern era of big data gives rise to situations in which we cannot assume that the number of observations is much larger than the dimension of the unknown parameter. It is particularly true for the distribution estimation problem, e.g., the Wikipedia page on the Chinese characters showed that the number of Chinese sinograms is at least $80,000$. Meanwhile, for distributions with extremely large support sizes (such as the Chinese language), the number of frequent symbols are considerably smaller than the support size. This observation motivates the third regime, in which we focus on distributions with finite entropy, but possibly extremely large support sizes. Another key result that motivates the problem of discrete distribution estimation under bounded entropy constraint is Marton and Shields \cite{Marton1994entropy}, who essentially showed that the entropy rate dictates the difficulty in estimating discrete distributions under $\ell_1$ loss in stochastic processes.

We denote by $\mathcal{M}_S$ the set of all distributions of support size $S$. The $\ell_1$ loss for estimating $P$ using $Q$ is defined as
\begin{align}
  \|P-Q\|_1\triangleq \sum_{i=1}^S |p_i-q_i|,
\end{align}
where $Q$ is not necessarily a probability mass function. The risk function for an estimator $\hat{P}$ in estimating $P$ under $\ell_1$ loss is defined as
\begin{equation}
R(P; \hat{P}) \triangleq \bE_P \|\hat{P}-P\|_1,
\end{equation}
where the expectation is taken with respect to the measure $P$. The maximum $\ell_1$ risk of an estimator $\hat{P}$, and the minimax risk in estimating $P$ are respectively defined as
\begin{align}
R_{\text{maximum}}(\mathcal{P};\hat{P}) & \triangleq \sup_{P \in \mathcal{P}} R(P; \hat{P}) \\
R_{\text{minimax}}(\mathcal{P}) & \triangleq \inf_{\hat{P}} \sup_{P \in \mathcal{P}} R(P; \hat{P}),
\end{align}
where $\mathcal{P}$ is a given collection of probability measures $P$, and the infimum is taken over all estimators $\hat{P}$.

This paper is dedicated to investigating the maximum risk of the MLE $R_{\text{maximum}}(\mathcal{P};P_n)$ and the minimax risk $R_{\text{minimax}}(\mathcal{P})$ for various $\mathcal{P}$. There are good reasons for focusing on the $\ell_1$ loss, as we do in this paper. Other loss functions in distribution estimation, such as the squared error loss, have been extensively studied in a series of papers~\cite{Steinhaus1957problem,Trybula1958problem,Rutkowska1977minimax,Olkin1979admissible}, while less is known for the $\ell_1$ loss. For the squared error loss, the minimax estimator is unique and depends on the support size $S$ \cite[Pg. 349]{Lehmann1998theory}. Since the support size $S$ is unknown in our setting, this estimator is highly impractical. This fact partially motivates our focus on the $\ell_1$ loss, which turns out to bridge our understanding of both parametric and nonparametric models. The $\ell_1$ loss, being proportional to the variational distance, is often a more natural measure of discrepancy between distributions than the $\ell_2$ loss. Moreover, the $\ell_1$ loss in discrete distribution estimation is compatible with and is a degenerate case of the $L_1$ loss in density estimation, which is the only loss that satisfies certain natural properties \cite{Devroye1985nonparametric}.

All logarithms in this paper are assumed to be in the natural base.

\subsection{Main results}

We investigate the maximum risk of the MLE $R_{\text{maximum}}(\mathcal{P};P_n)$ and the minimax risk $R_{\text{minimax}}(\mathcal{P})$ in the aforementioned three different regimes separately.

Understanding $R_{\text{maximum}}(\mathcal{P};P_n) = \sup_{P \in \mathcal{P}} R(P;P_n)$ follows from an understanding of $R(P;P_n) = \bE_P \| P_n - P \|_1$. This problem can be decomposed into analyzing the Binomial mean absolute deviation defined as
\begin{equation}
\bE \Big | \frac{X}{n} - p \Big|,
\end{equation}
where $X\sim \mathsf{B}(n,p)$ follows a Binomial distribution. Complicated as it may seem, De Moivre obtained an explicit expression for this quantity. Diaconis and Zabell provided a nice historical account for De Moivre's discovery in \cite{Diaconis1991closed}.

Berend and Kontorovich \cite{Berend2013sharp} provided tight upper and lower bounds on the Binomial mean absolute deviation, and we summarize some key results in Lemma \ref{lem_bino} of the Appendix. A well-known result to recall first is the following.

\begin{theorem}\label{th_MLE_S}
The maximum $\ell_1$ risk of the empirical distribution $P_n$ satisfies
\begin{align}
  \sup_{P\in\mathcal{M}_S} \mathbb{E}_P\|P_n-P\|_1 \le \sqrt{\frac{S-1}{n}},
\end{align}
where $\mathcal{M}_S$ denotes the set of distributions with support size $S$.
\end{theorem}

In fact, a tighter upper bound on the worst-case $\ell_1$ risk of the empirical distribution $P_n$ has been obtained in \cite{kamath2015learning}:
\begin{align}\label{eq:mleupper_tight}
    \sup_{P\in\mathcal{M}_S} \mathbb{E}_P\|P_n-P\|_1 \le \sqrt{\frac{2(S-1)}{\pi n}} + \frac{2S^{\frac{1}{2}}(S-1)^{\frac{1}{4}}}{n^{\frac{3}{4}}}.
\end{align}

In the present paper we show that MLE is minimax rate-optimal in all the three regimes we consider, but possibly suboptimal in terms of constants in high dimensional and infinite dimensional settings.

\subsubsection{Classical Asymptotics}
The next corollary is an immediate result from Theorem \ref{th_MLE_S}.
\begin{corollary}\label{cor_MLE_sample_S}
The empirical distribution $P_n$ achieves the worst-case convergence rate $O(n^{-\frac{1}{2}})$. Specifically,
    \begin{align}
     \limsup_{n\to \infty} \sqrt{n}\cdot\sup_{P\in\mathcal{M}_S}\mathbb{E}_P\|P_n-P\|_1 \le \sqrt{S-1} < \infty.
\end{align}
\end{corollary}

Regarding the lower bound, the well-known H\'{a}jek-Le Cam local asymptotic minimax theorem \cite{Hajek1972local} (Theorem \ref{lem_LAM} in Appendix \ref{sec_app_A}) and corresponding achievability theorems~\cite[Lemma 8.14]{Vandervaart2000} show that the MLE is optimal (even in constants) in classical asymptotics. Concretely, one corollary of Theorem~\ref{lem_LAM} in the Appendix shows the following.

\begin{corollary}\label{cor_n_S}
Fixing $S\ge2$, we have
\begin{align}
  \liminf_{n\to \infty} \sqrt{n}\cdot\inf_{\hat{P}}\sup_{P\in\mathcal{M}_S}\mathbb{E}_P\|\hat{P}-P\|_1 \ge \sqrt{\frac{2(S-1)}{\pi}} > 0,
\end{align}
where the infimum is taken over all estimators.
\end{corollary}

Note that the combination of (\ref{eq:mleupper_tight}) and Corollary \ref{cor_n_S} actually yields that the MLE is asymptotically minimax, and
\begin{equation}
\begin{split}
  & \lim_{n\to\infty} \sqrt{n}\cdot\inf_{\hat{P}}\sup_{P\in\mathcal{M}_S}\mathbb{E}_P\|\hat{P}-P\|_1 \\
  = & \lim_{n\to\infty} \sqrt{n}\cdot\sup_{P\in\mathcal{M}_S}\mathbb{E}_P\|P_n-P\|_1 =  \sqrt{\frac{2(S-1)}{\pi}},
\end{split}
\end{equation}
where the infimum is taken over all estimators. This result was also proved in \cite{kamath2015learning} via a different approach to obtain the exact constant for the lower bound in the classical asymptotics setting.

\subsubsection{High-dimensional Asymptotics}
Theorem \ref{th_MLE_S} also implies the following:
\begin{corollary}\label{cor_MLE_rate_S}
For $S=n/c$, the empirical distribution $P_n$ achieves the worst-case convergence rate $O(c^{-\frac{1}{2}})$, i.e.,
    \begin{align}
  \limsup_{c\to\infty}\sqrt{c}\cdot\limsup_{n\to \infty} \sup_{P\in\mathcal{M}_S}\mathbb{E}_P\|P_n-P\|_1 \le 1 < \infty.
\end{align}
\end{corollary}

Now we show that $S=n/c$ is the \emph{critical scaling} in high dimensional asymptotics. In other words, if $n = o(S)$, then no estimator for the distribution $P$ is consistent under $\ell_1$ loss. This phenomenon has been observed in several papers, such as \cite{Berend2013sharp} and \cite{Diakonikolas2014beyond}, to name a few.

The following theorem presents a non-asymptotic minimax lower bound.

\begin{theorem}\label{th_lower_S}
  For any $\zeta\in(0,1]$, we have
\begin{align}
  &\inf_{\hat{P}}\sup_{P\in\mathcal{M}_S} \bE_P\|\hat{P}-P\|_1 \ge \frac{1}{8}\sqrt{\frac{eS}{(1+\zeta)n}}\mathbbm{1}\left(\frac{(1+\zeta)n}{S}> \frac{e}{16}\right) \nonumber\\
  & \qquad\qquad + \exp\left(-\frac{2(1+\zeta)n}{S}\right)\mathbbm{1}\left(\frac{(1+\zeta)n}{S}\le \frac{e}{16}\right)\nonumber \\
  & \qquad\qquad - \exp(-\frac{\zeta^2n}{24}) - 12\exp\left(-\frac{\zeta^2S}{32(\ln S)^2}\right),
\end{align}
where the infimum is taken over all estimators.
\end{theorem}

Theorem~\ref{th_lower_S} implies the following minimax lower bound in high dimensional asymptotics, if we take $\zeta\to0^+$.
\begin{corollary}\label{cor_lower_sample_S}
For any constant $c>0$, if $S=n/c$, the convergence rate of the maximum $\ell_1$ risk is $\Omega(c^{-\frac{1}{2}})$. Specifically,
\begin{align}
  \liminf_{c\to\infty}\sqrt{c}\cdot\liminf_{n\to \infty} \inf_{\hat{P}}\sup_{P\in\mathcal{M}_S}\mathbb{E}_P\|\hat{P}-P\|_1 \ge \frac{\sqrt{e}}{8} > 0,
\end{align}
where the infimum is taken over all estimators.
\end{corollary}

Corollaries~\ref{cor_MLE_rate_S} and \ref{cor_lower_sample_S} imply that MLE achieves the optimal convergence rate $\Theta(c^{-\frac{1}{2}})$ in high dimensional linear scaling.

\subsubsection{Infinite-dimensional Asymptotics}
The performance of MLE in the regime of bounded entropy is characterized in the following theorem.
\begin{theorem}\label{th_MLE_H}
  The empirical distribution $P_n$ satisfies that, for any $H>0$ and $\eta>1$,
    \begin{align}
       \sup_{P: H(P)\le H} \mathbb{E}_P\|P_n-P\|_1 \le \frac{2H}{\ln n-2\eta\ln\ln n} + \frac{1}{(\ln n)^{\eta}}.
\end{align}
Further, for any $c\in(0,1)$ and $n>\max\{(1-c)^{-\frac{1}{1-c}},e^{H}\}$,
    \begin{align}
       \sup_{P: H(P)\le H} \mathbb{E}_P\|P_n-P\|_1 \ge \frac{2cH}{\ln n}\left(1-\left((1-c)n\right)^{-\frac{1}{c}}\right)^n.
\end{align}
\end{theorem}

The next corollary follows from Theorem~\ref{th_MLE_H} after taking $c\to1^-$.
\begin{corollary}\label{cor_rate_H}
For any $H>0$, the MLE $P_n$ satisfies
\begin{align}
  \lim_{n\to\infty} \frac{\ln n}{H}\cdot \sup_{P: H(P)\le H}\mathbb{E}_P\|P_n-P\|_1 = 2.
\end{align}
\end{corollary}
 It implies that we not only have obtained the $\Theta((\ln n)^{-1})$ convergence rate of the asymptotic $\ell_1$ risk of MLE, but also shown that the multiplicative constant is exactly $2H$. We note that this logarithmic convergence rate is really slow, implying that the sample size needs to be squared to reduce the maximum $\ell_1$ risk by a half. Also note that the maximum $\ell_1$ risk is proportional to the entropy $H$, thus the smaller the entropy of a distribution is known to be, the easier it is to estimate.

 However, given this slow rate $\Theta((\ln n)^{-1})$, it is of utmost importance to obtain estimators such that the corresponding multiplicative constant is as small as possible. We show that MLE does not achieve the optimal constant. In the following theorem, an essentially minimax estimator is explicitly constructed.
\begin{theorem}\label{th_lower_H}
 For any $\eta>1$, denote
\begin{align}
  \Delta_n \triangleq \frac{(\ln n)^{2\eta}}{n},
\end{align}
then for the hard-thresholding estimator defined as $\hat{P}^*(\mathbf{X}) = (g_n(X_1),g_n(X_2),\cdots,g_n(X_S))$ with
\begin{align}
g_n(X_i) = \frac{X_i}{n}\mathbbm{1}\left(\frac{X_i}{n}> e^2\Delta_n\right),
\end{align}
we have
\begin{align}\label{eq:est_optimal}
   &\sup_{P: H(P)\le H} \mathbb{E}_P\|\hat{P}^*-P\|_1 \le \frac{H}{\ln n-\ln(2e^2)-2\eta\ln\ln n} \nonumber\\
   &\qquad\qquad + (\ln n)^{-\eta} + n^{1-\frac{e^2}{4}}.
\end{align}
Moreover, for any $c\in(0,1)$ and $n\ge e^{H}$, we have
  \begin{align}
 \inf_{\hat{P}}\sup_{P: H(P)\le H} \mathbb{E}_P\|\hat{P}-P\|_1 \ge  \frac{cH}{\ln n}\cdot \left(1-n^{1-\frac{1}{c}}(1-c)^{-\frac{1}{c}}\right),
\end{align}
where the infimum is taken over all estimators.
\end{theorem}

Theorem \ref{th_lower_H} presents both a non-asymptotic achievable maximum $\ell_1$ risk and a non-asymptotic lower bound of the minimax $\ell_1$ risk, and it is straightforward to verify that the upper bound and lower bound coincide asymptotically by choosing $c\to1^-$. As a result, the asymptotic minimax $\ell_1$ risk is characterized in the following corollary.

\begin{corollary}\label{cor_lower_H}
For any $H>0$, the asymptotic minimax risk is $\frac{H}{\ln n}$, i.e.,
\begin{align}
  \lim_{n\to\infty} \frac{\ln n}{H}\cdot \inf_{\hat{P}}\sup_{P: H(P)\le H}\mathbb{E}_P\|\hat{P}-P\|_1 = 1,
\end{align}
and the estimator $\hat{P}$ in Theorem \ref{th_lower_H} is asymptotically minimax.
\end{corollary}

In light of Corollaries \ref{cor_rate_H} and~\ref{cor_lower_H}, the asymptotic minimax $\ell_1$ risk for the distribution estimation with bounded entropy is exactly $\frac{H}{\ln n}$, half of that obtained by MLE. Since the convergence rate is $\Theta((\ln n)^{-1})$, the performance of the asymptotically minimax estimator with $n$ samples in this problem is nearly that of MLE with $n^2$ samples, which is a significant improvement.

Note that the asymptotically minimax estimator given in Theorem \ref{th_lower_H} is a hard-thresholding estimator which neglects all symbols with frequency less than $e^2\Delta_n=e^2(\ln n)^{2\eta}/n$, and its risk depends on $\eta$ through two factors. On one hand, there is a loss due to ignoring low frequency symbols, which increases with $\eta$ and corresponds to the first term in the right-hand side of (\ref{eq:est_optimal}). On the other hand, the risk incurred by the involvement of high frequency symbols decreases with $\eta$ and corresponds to the $(\ln n)^{-\eta}$ term in (\ref{eq:est_optimal}). Taking derivative with respect to $\eta$ yields that the optimal parameter $\eta^*$ to handle this trade-off takes the form $\eta^*=c_n\ln n/\ln\ln n$ with coefficients $c_n\to 0$ (but larger than $\Theta(\ln \ln n/\ln n)$), and then the upper bound of the $\ell_1$ risk in (\ref{eq:est_optimal}) becomes
\begin{align}
  &\sup_{P: H(P)\le H} \mathbb{E}_P\|\hat{P}^*-P\|_1 \le \frac{H}{(1-2c_n)\ln n-\ln(2e^2)} \nonumber \\
   & \qquad \qquad + n^{-c_n} + n^{1-\frac{e^2}{4}}.
\end{align}

However, one may notice that the asymptotically minimax estimator in Theorem \ref{th_lower_H} outputs estimates that are not necessarily probability distributions. What if we constrain the estimator $\hat{P}$ to output estimates confined to $\mathcal{M}_\infty\triangleq \cup_{S=1}^\infty \mathcal{M}_S$, i.e., are bona fide probabilities? In this case, the next theorem shows that the MLE is asymptotically minimax again.

\begin{theorem}\label{th_lower_H_mle}
  For any $c\in(0,1)$ and $n\ge e^{H}$, we have
  \begin{align}
 &\inf_{\hat{P}\in\mathcal{M}_\infty}\sup_{P: H(P)\le H} \mathbb{E}_P\|\hat{P}-P\|_1 \nonumber\\
 &\qquad\qquad\ge  \frac{2cH}{\ln n}\cdot \left(1-n^{1-\frac{1}{c}}(1-c)^{-\frac{1}{c}}\right),
\end{align}
where the infimum is taken over all estimators with outputs confined to the probability simplex.
\end{theorem}

The next corollary follows immediately from Corollary \ref{cor_rate_H} upon taking $c\to 1^-$:
\begin{corollary}\label{cor_lower_H_mle}
For any $H>0$, the asymptotic minimax risk is $\frac{2H}{\ln n}$ when the estimates are restricted to the probability simplex:
\begin{align}
  \lim_{n\to\infty} \frac{\ln n}{H}\cdot \inf_{\hat{P}\in\mathcal{M}_\infty}\sup_{P: H(P)\le H}\mathbb{E}_P\|\hat{P}-P\|_1 = 2,
\end{align}
and the MLE is asymptotically minimax.
\end{corollary}

Hence, the MLE still enjoys the asymptotically minimax property in the bounded entropy case if the estimates have to be probability mass functions. We can explain this peculiar phenomenon from the Bayes viewpoint as follows. By the minimax theorem \cite{Wald1950statistical}, any minimax estimator can be always approached by a sequence of Bayes estimators under some priors. However, for general loss functions including the $\ell_1$ loss, the corresponding Bayes estimator (without any restrictions) may not belong to the probability simplex $\mathcal{M}_\infty$ even if the prior is supported on $\mathcal{M}_\infty$. Hence, the space of all Bayes estimators which form probability masses is strictly smaller than the space of all free Bayes estimators without any constraints. For example, it is well-known that under the $\ell_1$ loss, the Bayes response is actually the median vector of the posterior distribution, and thus does not necessarily form a probability distribution. We also remark that for some special loss functions such as the squared error loss, the Bayes estimator under any prior supported on $\mathcal{M}_\infty$ will always belong to $\mathcal{M}_\infty$. Broadly speaking, if the loss function is a Bregman divergence \cite{banerjee2005clustering}, the Bayes estimator under any prior will always be the conditional expectation, which is supported on the convex hull of the parameter space.

\subsection{Discussion}

Now we draw various connections between our results and the literature.

\subsubsection{Density estimation under $L_1$ loss}
There is an extensive literature on density estimation under $L_1$ loss, and we refer to the book by Devroye and Gyorfi~\cite{Devroye1985nonparametric} for an excellent overview. This problem is also very popular in theoretical computer science, e.g. \cite{Chan2014Efficient}.

The problem of discrete distribution estimation under $\ell_1$ loss has been the subject of recent interest in the theoretical computer science community. For example, Daskalakis, Diakonikolas and Servedio considered the problem of $k$-modal distributions \cite{daskalakis2012learning}, and a very recent talk by Diakonikolas \cite{Diakonikolas2014beyond} provided a literature survey. The conclusion that it is necessary and sufficient to use $n=\Theta(S)$ samples to consistently estimate an arbitrary discrete distribution with support size $S$ has essentially appeared in the literature \cite{Diakonikolas2014beyond}, but we did not find an explicit reference giving non-asymptotic results, and for completeness we have included proofs corresponding to the high dimensional asymptotics in this paper. We remark that, a very detailed analysis of the discrete distribution estimation problem under $\ell_1$ loss may be insightful and instrumental in future breakthroughs in density estimation under $L_1$ loss.

\subsubsection{Entropy estimation}

It was shown that the entropy is nearly Lipschitz continuous under the $\ell_1$ norm \cite[Thm. 17.3.3]{Cover1991information}\cite[Lemma 2.7]{csiszar1981information}, i.e., if $\|P-Q\|_1\le1/2$, then
\begin{align}\label{eq:abs}
  |H(P)-H(Q)|\le -\|P-Q\|_1 \ln\frac{\|P-Q\|_1}{S},
\end{align}
where $S$ is the support size. At first glance, it seems to suggest that the estimation of entropy can be reduced to estimation of discrete distributions under $\ell_1$ loss. However, this question is far more complicated than it appears.

First, people have already noticed that this near-Lipschitz continuity result is valid only for finite alphabets \cite{Silve--Parada2012Shannon}. The following result by Antos and Kontoyiannis \cite{Antos2001convergence} addresses entropy estimation over countably infinite support sizes.
\begin{remark}\label{rem_antos}
Among all discrete sources with finite entropy and $\mathsf{Var}(-\ln p(X))<\infty$, for any sequence $\{H_n\}$ of estimators for the entropy, and for any sequence $\{a_n\}$ of positive numbers converging to zero, there is a distribution $P$ (supported on at most countably infinite symbols) with $H=H(P)<\infty$ such that
\begin{align}
  \limsup_{n\to\infty} \frac{\bE_P|H_n-H|}{a_n} = \infty.
\end{align}
\end{remark}

Remark \ref{rem_antos} shows that, among all sources with finite entropy and finite varentropy, no rate-of-convergence results can be obtained for any sequence of estimators. Indeed, if the entropy is still nearly-Lipschitz continuous with respect to $\ell_1$ distance in the infinite alphabet setting, then Corollary~\ref{cor_rate_H} immediately implies that the MLE plug-in estimator for entropy attains a universal convergence rate. That there is no universal convergence rate of entropy estimators for sources with bounded entropy is particularly interesting in light of the fact that the minimax rates of convergence of distribution estimation with bounded entropy is $\Theta((\ln n)^{-1})$.

Second, it is very interesting and deserves pondering that along high dimensional asymptotics, the minimax sample complexity for estimating entropy is $n=\Theta(\frac{S}{\ln S})$ samples, a result first discovered by Valiant and Valiant \cite{Valiant--Valiant2011}, then recovered by Jiao \emph{et al.} using a different approach in \cite{jiao2015minimax}, and Wu and Yang in \cite{Wu--Yang2014minimax}. Since it is shown in Corollary \ref{cor_lower_sample_S} that we need $n=\Theta(S)$ samples to consistently estimate the distribution, this result shows that we can consistently estimate the entropy without being able to consistently estimate the underlying distribution under $L_1$ loss. Note that if the plug-in approach is used for entropy estimation, i.e., if we use $H(P_n)$ to estimate $H(P)$, it has been shown in \cite{paninski2003estimation,jiao2014nonasymptotic} that this estimator again requires $n=\Theta(S)$ samples. In fact, for a wide class of functionals of discrete distributions, it is shown in \cite{jiao2015minimax} that the MLE is strictly suboptimal, and the performance of the optimal estimators with $n$ samples is essentially that of the MLE with $n\ln n$ samples, a phenomenon termed ``effective sample size enlargement" in \cite{jiao2015minimax}. Jiao \emph{et al.}\cite{jiao2014beyond} showed that the improved estimators introduced in~\cite{jiao2015minimax} can lead to consistent and substantial performance boosts in various machine learning algorithms.

\subsubsection{$\ell_1$ divergence estimation between two distributions}

Now we turn to the estimation problem for the $\ell_1$ divergence $\|P-Q\|_1$ between two discrete distributions $P,Q$ with support size at most $S$. At first glance, by setting one of the distributions to be deterministic, the problem of $\ell_1$ divergence estimation seems a perfect dual to the distribution estimation problem under $\ell_1$ loss. However, compared to the required sample complexity $n=\Theta(S)$ in the distribution estimation problem under $\ell_1$ loss, the minimax sample complexity for estimating the $\ell_1$ divergence between two arbitrary distributions is $n=\Theta(\frac{S}{\ln S})$ samples, a result of Valiant and Valiant \cite{Valiant--Valiant2011power}. Hence, it is easier to estimate the $\ell_1$ divergence than to estimate the distribution with a vanishing $\ell_1$ risk. Note that for distribution estimation, for each symbol we need to obtain a good estimate for $p_i$ in terms of the $\ell_1$ risk, while for $\ell_1$ divergence estimation we do not need to estimate each $p_i$ and $q_i$ separately.

\subsubsection{Joint $d$-block distribution estimation in stochastic processes}

Insights can be gleaned from the comparison of Corollary \ref{cor_rate_H} and the result obtained by Marton and Shields \cite{Marton1994entropy}. Marton and Shields showed that, in a stationary ergodic stochastic process with sample size $n$ and entropy rate $H$, the joint $d$-tuple distribution of the process can be consistently estimated using the empirical distribution if $d\le \frac{(1-\epsilon)\ln n}{H}$, where $\epsilon>0$ is an arbitrary constant. Moreover, the empirical distribution is not consistent if $d\geq  \frac{(1+\epsilon)\ln n}{H}$. Now we treat the joint $d$-tuple distribution as a single distribution with a large support size, and consider the corresponding estimation problem. To be precise, we assume without loss of generality that the original process consists of $n$ observations and merge all disjoint blocks containing $d=\frac{(1-\epsilon)\ln n}{H}$ symbols into supersymbols. Consequently, we obtain a sample size of $n/d$ from which we would like to learn a new distribution over an alphabet of size $S^d$ and entropy nearly $dH$. Considering a special case where the stationary ergodic process is nearly i.i.d. and applying our result on the minimum sample complexity of the distribution estimation problem with bounded support size, we need $\Theta(S^d)=\Theta(n^{\frac{(1-\epsilon)\ln S}{H}})$ samples to estimate the new distribution, which cannot be achieved by $n$ samples unless $\ln S=H$, i.e., the distribution is uniform. Furthermore, under the regime of distribution estimation with bounded entropy, the asymptotic minimax risk of MLE will be $\frac{2dH}{\ln n}=2(1-\epsilon)$, which does not vanish as $n\to\infty$. Hence, we conclude that the minimax distribution estimation under i.i.d. samples is indeed harder than the joint $d$-tuple distribution estimation in a stationary ergodic process. To clarify the distinction, we remark that the ergodicity plays a crucial role in the latter case: for $d\to\infty$, the Shannon-McMillan-Breiman theorem \cite{Algoet--Cover1988sandwich} guarantees that there are approximately $e^{dH}$ typical sequences of length $d$, each of which occurs with probability about $e^{-dH}$. Then it turns out that we need only to estimate a uniform distribution with support size $e^{dH}$, and applying the previous conclusion $n/d=\Theta(e^{dH})$ yields the desired result $d\approx\frac{\ln n}{H}$. On the other hand, for the distribution estimation problem with a large support size, the equipartition property does not necessarily hold, and our scheme focuses on the worst-case risk over all possible distributions. Hence, it is indeed harder to handle the distribution estimation problem under i.i.d. samples in the large alphabet regime, and it indicates that directly applying results from a large alphabet regime to stochastic processes with a large memory may not be a fruitful route. Before closing the discussion, we mention that Marton and Shields did not show that the $d \sim \frac{\ln n}{H}$ scaling is minimax optimal. It remains an interesting question to investigate whether there is a better estimator to estimate the $d$-tuple joint distribution in stochastic processes.

\subsubsection{Hard-thresholding estimation is asymptotically minimax}

Corollary \ref{cor_lower_H} shows that in the infinite dimensional asymptotics, MLE is far from asymptotically minimax, and a hard-thresholding estimator achieves the asymptotic minimax risk. The phenomenon that thresholding methods are needed in order to obtain minimax estimators for high dimensional parameters in a $\ell_p$ ball under $\ell_q$ error, $p>0, p<q, q\in[1,\infty)$, was first noticed by Donoho and Johnstone \cite{Donoho--Johnstone1994minimax}. Following the rationale of the James-Stein shrinkage estimator \cite{Stein1956inadmissibility}, Donoho and Johnstone proposed the soft- and hard-thresholding estimators for the normal mean given that we know \emph{a priori} that the mean $\theta$ lies in a $\ell_p$ ball, $p\in(0,\infty)$. Later, Donoho and Johnstone applied this idea to nonparametric estimation in Besov spaces, and obtained the famous \emph{wavelet shrinkage} estimator for denoising \cite{Donoho--Johnstone1994ideal}. Note that the set $\{P: H(P)\le H\}$ forms a ball similar to the $\ell_p$ ball, and the loss function is $\ell_1$, so it is not surprising that hard-thresholding leads to an asymptotically minimax estimator. The asymptotic minimax estimators under other constraints on the distribution $P$ remain to be explored.

\subsubsection{Adaptive estimation}

Note that in the infinite dimensional asymptotics, for a sequence of problems $\{H(P)\le H\}$ with different upper bounds $H$, the asymptotically minimax estimator in Theorem~\ref{th_lower_H} achieves the minimax risk over all ``entropy balls" without knowing its ``radius" $H$. It is very important in practice, since we do not know a priori an upper bound on the entropy of the distribution. This estimator belongs to a general collections of estimators called adaptive estimators. For details we refer to a survey paper by Cai \cite{cai2012}.

The rest of this paper is organized as follows. Section \ref{sec_2} provides outlines of the proofs of the main theorems, and some useful auxiliary lemmas are listed in Appendix \ref{sec_app_A}. Complete proofs of some lemmas and corollaries are provided in Appendix \ref{sec_app_B}.

\newcounter{mytempeqncnt}
\begin{figure*}[!b]
\normalsize
\setcounter{mytempeqncnt}{\value{equation}}
\setcounter{equation}{53}
\hrulefill
\begin{align}
  \mathbb{E}_P\|P-\hat{P}\|_1 &\le \sum_{p_i\le \Delta_n}\left(p_i+\left(\frac{p_i}{e\Delta_n}\right)^{e^2\ln n}\right)
  + \sum_{\Delta_n<p_i<2e^2\Delta_n}\left(\sqrt{\frac{p_i}{n}} + p_i\right)
  + \sum_{p_i\ge2e^2\Delta_n}\left(\sqrt{\frac{p_i}{n}} + n^{-\frac{e^2}{4}}\right)\\
  &= \sum_{p_i< 2e^2\Delta_n}p_i + \sum_{p_i>\Delta_n}\sqrt{\frac{p_i}{n}} + \sum_{p_i\le \Delta_n}\left(\frac{p_i}{e\Delta_n}\right)^{e^2\ln n} + \sum_{p_i\ge2e^2\Delta_n}n^{-\frac{e^2}{4}}\\
  &\le \left(\ln \frac{1}{2e^2\Delta_n}\right)^{-1}\cdot H + \frac{1}{\sqrt{n\Delta_n}} + e^{-e^2\ln n}\cdot \frac{n}{(\ln n)^{2\eta}} + n^{-\frac{e^2}{4}}\cdot \frac{1}{2e^2\Delta_n}\\
  &= \frac{H}{\ln n-\ln(2e^2)-2\eta\ln\ln n} + \frac{1}{(\ln n)^{\eta}} + \frac{1}{n^{e^2-1}(\ln n)^{2\eta}} + \frac{1}{2e^2n^{\frac{e^2}{4}-1}(\ln n)^{2\eta}}\\
  &\le \frac{H}{\ln n-\ln(2e^2)-2\eta\ln\ln n} + (\ln n)^{-\eta} + n^{1-\frac{e^2}{4}}
\end{align}
\setcounter{equation}{\value{mytempeqncnt}}
\vspace*{4pt}
\end{figure*}
\setcounter{equation}{27}

\section{Outlines of Proofs of Main Theorems}\label{sec_2}
\subsection{Analysis of MLE}
For the analysis of the performance of the MLE $P_n$, the key is to obtain a good approximation of $\bE\left|X/n-p\right|$ with $X\sim \mathsf{B}(n,p)$, i.e., the Binomial mean absolute deviation. Lemma \ref{lem_bino} in the Appendix lists some sharp approximations, which together with the concavity of $\sqrt{x(1-x)}$ yield
\begin{align}
  \mathbb{E}_P\|P-P_n\|_1 &= \sum_{i=1}^S \bE_P\left|\frac{X_i}{n}-p_i\right| \\
  &\le \sum_{i=1}^S \sqrt{\frac{p_i(1-p_i)}{n}} \\
  &\le \sqrt{\frac{S-1}{n}},
\end{align}
which completes the proof of Theorem \ref{th_MLE_S}. For the upper bound in Theorem \ref{th_MLE_H}, we use Lemma \ref{lem_bino} again and obtain
\begin{align}
  \mathbb{E}_P\|P-P_n\|_1
&\le \sum_{i=1}^S \min\left\{\sqrt{\frac{p_i}{n}},2p_i\right\}\\
  &\le 2\sum_{p_i\le \frac{(\ln n)^{2\eta}}{n}} p_i + \frac{1}{\sqrt{n}}\sum_{p_i> \frac{(\ln n)^{2\eta}}{n}} \sqrt{p_i}\\
  &\le \frac{2}{\ln n - 2\eta\ln\ln n}\sum_{p_i\le \frac{(\ln n)^{2\eta}}{n}} \left(-p_i\ln p_i\right) \nonumber\\
   & \qquad + \frac{1}{\sqrt{n}}\cdot \sqrt{\left|\left\{i: p_i> \frac{(\ln n)^{2\eta}}{n}\right\}\right|}\\
  &\le \frac{2H}{\ln n-2\eta\ln\ln n} + \frac{1}{(\ln n)^{\eta}}.
\end{align}

For the lower bound, we consider the distribution $P=(\delta/S',\cdots,\delta/S',1-\delta)$ with entropy $H$, then for $c\in(0,1),\delta\le c$, due to the monotone decreasing property of $(1-x)^{\frac{1}{x}}$ with respect to $x\in(0,1)$,
\begin{align}
  S' &= \delta\exp\left(\frac{H}{\delta}\right)\cdot \left((1-\delta)^{\frac{1}{\delta}}\right)^{1-\delta} \\
  &\ge \delta\exp\left(\frac{H}{\delta}\right)\cdot \left(1-c\right)^{\frac{1}{c}}.
\end{align}
Note that $S'+1$ is the support size, and we assume without loss of generality that $S'$ is an integer. For $c\in(0,1)$, since $n\ge e^H$, we choose $\delta=cH/\ln n\le c$, then since $n>(1-c)^{-\frac{1}{1-c}}$,
\begin{align}
  \frac{\delta}{S'} \le (1-c)^{-\frac{1}{c}}\exp\left(-\frac{H}{\delta}\right) = \left((1-c)n\right)^{-\frac{1}{c}} < \frac{1}{n},
\end{align}
and the identity in Lemma \ref{lem_bino} can be applied to obtain
\begin{align}
  \mathbb{E}_P\|P-P_n\|_1
  &\ge \sum_{i=1}^{S'} \mathbb{E}_P \left|p_i-\frac{X_i}{n}\right| \\
  &= 2\delta \left(1-\frac{\delta}{S'}\right)^n \\
  &\ge \frac{2cH}{\ln n}\left(1-((1-c)n)^{-\frac{1}{c}}\right)^n.
\end{align}

\begin{figure*}[!b]
\normalsize
\setcounter{mytempeqncnt}{\value{equation}}
\setcounter{equation}{58}
\hrulefill
\begin{align}
  R_B(S,n,\mu_0^S) &= \int \mathbb{E}_P\|P-\hat{P}^B\|_1 \mu(dP)\\
  &= \sum_{i=1}^{S}\int \sum_{k=0}^\infty |p_i-f(k)|e^{-np_i}\frac{(np_i)^k}{k!} \mu_0(dp_i)\\
  &\ge \sum_{i=1}^{S}\int\sum_{k=0}^\infty |p_i-f(k)|\min\left\{\bP\{\mathsf{Poi}\left(\frac{n(1-\eta)}{S}\right)=k\} , \bP\{\mathsf{Poi}\left(\frac{n(1+\eta)}{S}\right)=k\}\right\} \mu_0(dp_i)\\
  &\ge S\cdot \frac{\eta}{S}\sum_{k=0}^\infty \min\left\{\bP\{\mathsf{Poi}\left(\frac{n(1-\eta)}{S}\right)=k\} , \bP\{\mathsf{Poi}\left(\frac{n(1+\eta)}{S}\right)=k\}\right\}\\
  &= \eta \cdot\sum_{k=0}^\infty\min\left\{\bP\{\mathsf{Poi}\left(\frac{n(1-\eta)}{S}\right)=k\} , \bP\{\mathsf{Poi}\left(\frac{n(1+\eta)}{S}\right)=k\}\right\}\\
  &= \eta \cdot\left(1-d_{\text{TV}}\left(\mathsf{Poi}\left(\frac{n(1-\eta)}{S}\right),\mathsf{Poi}\left(\frac{n(1+\eta)}{S}\right)\right)\right),
\end{align}
\setcounter{equation}{\value{mytempeqncnt}}
\vspace*{4pt}
\end{figure*}
\setcounter{equation}{40}

\subsection{Analysis of the Estimator in Theorem \ref{th_lower_H}}
For the achievability result, we first establish some lemmas on the properties of $g_n(X)$.
\begin{lemma}
  If $X\sim \mathsf{B}(n,p),p\le \Delta_n$, we have
  \begin{align}
     \bE \left|g_n(X)-p\right|\le p+\left(\frac{p}{e\Delta_n}\right)^{e^2\ln n}.
  \end{align}
\end{lemma}
\begin{proof}
  It is clear from the triangle inequality that
  \begin{align}
    \bE \left|g_n(X)-p\right| &\le p + \bE\left[\frac{X}{n}\mathbbm{1}\left(\frac{X}{n}>e^2\Delta_n\right)\right]\\
    &\le p + \bP(X>e^2n\Delta_n)\\
    &\le p+\left(\frac{np}{en\Delta_n}\right)^{e^2(\ln n)^{2\eta}}
  \end{align}
 where Lemma \ref{lem_chernoff} in Appendix A is used in the last step. The proof is completed by noticing that $\frac{p}{e\Delta_n}\le \frac{1}{e}<1$.
\end{proof}

\begin{lemma}
  If $X\sim \mathsf{B}(n,p),p\ge 2e^2\Delta_n$, we have
  \begin{align}
     \bE \left|g_n(X)-p\right| \le  \sqrt{\frac{p}{n}} + n^{-\frac{e^2}{4}}.
  \end{align}
\end{lemma}
\begin{proof}
  It follows from the identity
  \begin{align}
    g_n(X) = \frac{X}{n} - \frac{X}{n}\mathbbm{1}\left(\frac{X}{n}\le e^2\Delta_n\right)
  \end{align}
  and the triangle inequality that
  \begin{align}
    \bE \left|g_n(X)-p\right| &\le \bE \left|\frac{X}{n}-p\right| + \bE\left[\frac{X}{n}\mathbbm{1}\left(\frac{X}{n}\le e^2\Delta_n\right)\right]\\
    &\le\bE \left|\frac{X}{n}-p\right| + \bP(X\le e^2n\Delta_n)\\
    &\le \sqrt{\frac{p}{n}} + e^{-\frac{e^2(\ln n)^{2\eta}}{4}}.
  \end{align}
  Then the proof is completed by noticing that $e^{-\frac{e^2(\ln n)^{2\eta}}{4}}\le e^{-\frac{e^2\ln n}{4}}=n^{-\frac{e^2}{4}}$.
\end{proof}

\begin{lemma}
  If $X\sim \mathsf{B}(n,p), \Delta_n <p< 2e^2\Delta_n$, we have
  \begin{align}
     \bE \left|g_n(X)-p\right| \le \sqrt{\frac{p}{n}} +p.
  \end{align}
\end{lemma}
\begin{proof}
  It is clear that
  \begin{align}
    \bE \left|g_n(X)-p\right| &\le \bE \left[\left|\frac{X}{n}-p\right|\mathbbm{1}\left(\frac{X}{n}> e^2\Delta_n\right)\right] \nonumber\\
    &\qquad + \bE \left[|0-p|\mathbbm{1}\left(\frac{X}{n}\le e^2\Delta_n\right)\right] \\
    & \le\bE \left|\frac{X}{n}-p\right| + p \\
    & \le \sqrt{\frac{p}{n}} + p.
\end{align}
\end{proof}

Combining these lemmas, the upper bound of Theorem \ref{th_lower_H} is given on the bottom of this page.

\subsection{Proof of the Lower Bounds in Theorem \ref{th_lower_S}, \ref{th_lower_H} and \ref{th_lower_H_mle}}
To obtain a lower bound for the minimax risk, an effective way is to use the Bayes risk to serve as the lower bound \cite{Lehmann1998theory}, where the prior can be arbitrarily chosen. Hence, our target is to find an unfavorable prior and compute the corresponding Bayes risk. For computational simplicity, in the proof we will assign the product of independent priors to the whole probability vector $P$ based on the Poissonized model $X_i\sim \mathsf{Poi}(np_i)$, and then use some concentration inequalities such as the Hoeffding bound to ensure that the vector $P$ is close to a probability distribution with overwhelming probability. Then the relationship between the minimax risk of the Poissonized model and that of the Multinomial model needs to be established. The rigorous proofs are detailed as follows.

\setcounter{equation}{64}
\subsubsection{Lower Bound in Theorem \ref{th_lower_S}}
We denote the uniform distribution on two points $\{\frac{1-\eta}{S},\frac{1+\eta}{S}\}$ by $\mu_0$, with $\eta\in(0,1)$ to be specified later, and assign the product measure $\mu_0^S$ to the probability vector $P$. Under the Poissonized model $X_i\sim \mathsf{Poi}(np_i)$, it is straightforward to see that all $p_i(1\le i\le S)$ are conditionally independent given $\mathbf{X}$. Hence, the Bayes estimator $\hat{P}^B(\bf{X})$ under prior $\mu_0^S$ can be decomposed into $\hat{P}^B({\bf X}) = (f(X_1),f(X_2),\cdots,f(X_S))$, for some function $f(\cdot)$. Then the Bayes risk is shown on the bottom of this page, where $d_{\text{TV}}(\mathsf{Poi}(u),\mathsf{Poi}(v))$ is the variational distance between two Poisson distributions:
\begin{align}
  & d_{\text{TV}}(\mathsf{Poi}(u),\mathsf{Poi}(v))\nonumber\\
   &\qquad \triangleq \sup_{A\subset\mathbb{N}} \left|\bP\left\{\mathsf{Poi}(u)\in A\right\}-\bP\left\{\mathsf{Poi}(v)\in A\right\}\right|.
\end{align}

Adell and Jodra \cite{Adell--Jodra2006Exact} gives an upper bound for this distance:
\begin{align}
  &d_{\text{TV}}(\mathsf{Poi}(t),\mathsf{Poi}(t+x)) \le \nonumber\\
  &\qquad \min\left\{1-e^{-x}, \sqrt{\frac{2}{e}}(\sqrt{t+x}-\sqrt{t})\right\}, \quad t,x\ge 0
\end{align}
then
\begin{align}
  &d_{\text{TV}}\left(\mathsf{Poi}\left(\frac{n(1-\eta)}{S}\right),\mathsf{Poi}\left(\frac{n(1+\eta)}{S}\right)\right)\\
  &\qquad\le \min\left\{1-\exp\left(-2\eta\cdot\frac{n}{S}\right), \sqrt{\frac{2n}{eS}}(\sqrt{1+\eta}-\sqrt{1-\eta})\right\}\\
  &\qquad\le \min\left\{1-\exp\left(-2\eta\cdot\frac{n}{S}\right), 2\eta\cdot\sqrt{\frac{n}{eS}}\right\}.
\end{align}

Hence, by setting
\begin{align}
  \eta = \min\left\{1, \frac{1}{4}\sqrt{\frac{eS}{n}}\right\}
\end{align}
we can obtain
\begin{align}\label{eq:R_B}
  R_B(S,n,\mu_0^S) \ge \begin{cases}
  \exp\left(-\frac{2n}{S}\right), &\frac{n}{S}\le \frac{e}{16} \\
  \frac{1}{8}\sqrt{\frac{eS}{n}}, &\frac{n}{S}> \frac{e}{16}.
\end{cases}
\end{align}

The combination of Lemma \ref{lem_relate} and \ref{lem_bayes} in Appendix A yields, for any $\zeta\in(0,1]$ and $\epsilon\in(0,\frac{\zeta}{2(1+\zeta)}]$,
\begin{align}\label{eq:R_B_1}
  &\inf_{\hat{P}}\sup_{P\in\mathcal{M}_S} \bE_P\|\hat{P}-P\|_1 \ge R_B(S,(1+\zeta)n,\mu_0^S) \nonumber \\
  &\qquad - \exp(-\frac{\zeta^2n}{24}) - 6\mu_0^S\left(\mathcal{M}_S(\epsilon)^c\right).
\end{align}
Setting $\epsilon=\frac{\zeta}{4\ln S}$, Lemma \ref{lem_hoeffding} in Appendix A yields
\begin{align}\label{eq:R_B_2}
  \mu_0^S\left(\mathcal{M}_S(\epsilon)^c\right) &= \mu_0^S\left\{\left|\sum_{i=1}^S p_i - \bE\left[\sum_{i=1}^S p_i\right]\right|\ge \epsilon\right\}\\
&\le 2\exp\left(-\frac{2\epsilon^2}{S\cdot(2/S)^2}\right) \\
&= 2\exp\left(-\frac{\zeta^2S}{32(\ln S)^2}\right).
\end{align}
The proof of Theorem \ref{th_lower_S} is completed by the combination of (\ref{eq:R_B}), (\ref{eq:R_B_1}) and (\ref{eq:R_B_2}).
\begin{figure*}[!b]
\normalsize
\setcounter{mytempeqncnt}{\value{equation}}
\setcounter{equation}{86}
\hrulefill
\begin{align}
  R_B({\bf X},H,n,\mu)&\triangleq \sum_{i=1}^{N({\bf X})} \left|\frac{\delta}{S'}-a_i\right| + \sum_{j=N({\bf X})+1}^{kS'} \left[\frac{(k-1)S'}{kS'-N({\bf X})}\left|0-a_j\right|+\frac{S'-N({\bf X})}{kS'-N({\bf X})}\left|\frac{\delta}{S'}-a_j\right|\right] + |1-\delta-a_S|\\
  &\ge \sum_{i=1}^{N({\bf X})} \left|\frac{\delta}{S'}-a_i\right| + \sum_{j=N({\bf X})+1}^{kS'} \left|\frac{(k-1)S'}{kS'-N({\bf X})}\cdot a_j+\frac{S'-N({\bf X})}{kS'-N({\bf X})}\cdot\left(\frac{\delta}{S'}-a_j\right)\right| + |1-\delta-a_S|\\
  &= \sum_{i=1}^{N({\bf X})} \left|\frac{\delta}{S'}-a_i\right| + \sum_{j=N({\bf X})+1}^{kS'} \left|\lambda a_j+\frac{\delta}{kS'-N({\bf X})}\cdot\left(1-\frac{N({\bf X})}{S'}\right)\right| + |1-\delta-a_S|\\
  &\ge \sum_{i=1}^{N({\bf X})} \lambda\left|a_i-\frac{\delta}{S'}\right| + \sum_{j=N({\bf X})+1}^{kS'} \left|\lambda a_j+\frac{\delta}{kS'-N({\bf X})}\cdot\left(1-\frac{N({\bf X})}{S'}\right)\right| + \lambda|a_S-1+\delta|\\
  &\ge \left|\sum_{i=1}^{N({\bf X})} \lambda\left(a_i-\frac{\delta}{S'}\right) + \sum_{j=N({\bf X})+1}^{kS'} \left(\lambda a_j+\frac{\delta}{kS'-N({\bf X})}\cdot\left(1-\frac{N({\bf X})}{S'}\right)\right) + \lambda(a_S-1+\delta) \right|\\
  &= \frac{2(k-1)S'}{kS'-N({\bf X})}\left(1-\frac{N({\bf X})}{S'}\right)\delta,
\end{align}
\setcounter{equation}{\value{mytempeqncnt}}
\vspace*{4pt}
\end{figure*}
\setcounter{equation}{75}
\subsubsection{Lower Bound in Theorem \ref{th_lower_H} and \ref{th_lower_H_mle}}
For the proof of the lower bound in Theorem \ref{th_lower_H}, we need a different prior. Specifically, we fix some $\delta\in(0,1)$ with its value to be specified later, and consider $S'$ with
\begin{align}\label{eq:S}
  \delta\ln S' - \delta\ln \delta - (1-\delta)\ln (1-\delta) = H.
\end{align}

Now define $S=kS'+1$ with parameter $k\ge 2$ to be specified later, and consider the following collection $\mathcal{N}_{S,H}$ of $S$-dimensional non-negative vectors: $P\in\mathcal{N}_{S,H}$ if and only if $p_{S}=1-\delta,p_i\in\{0,\frac{\delta}{S'}\},1\le i\le S-1=kS'$, and
\begin{align}
  \left|\left\{i: 1\le i\le kS', p_i = \frac{\delta}{S'}\right\}\right| = S'.
\end{align}

By the preceding two equalities, we can easily verify that $P\in\mathcal{N}_{S,H}$ implies $H(P)=H$. Now we assign the uniform distribution $\mu$ on $\mathcal{N}_{S,H}$ to the distribution vector $P$, and consider the Bayes estimator $\hat{P}^B$ of $P$ under the prior $\mu$. We first compute the posterior distribution of $P$ given an observation vector $\bf X$. Denote by $N({\bf X})$ the number of different symbols (excluding the last symbol) appearing in the observation vector $\bf X$, and we assume without loss of generality that the first $N({\bf X})$ symbols appear in the sample. Then it is straightforward to show that the posterior distribution of $P$ on $\bf X$ is uniform in the following set $\mathcal{N}_{\bf X}\subset \mathcal{N}_{S,H}$: $P\in\mathcal{N}_{\bf X}$ if and only if $p_S=1-\delta, p_i=\frac{\delta}{S'}, 1\le i\le N({\bf X}), p_j\in\{0,\frac{\delta}{S'}\},N({\bf X})+1\le j\le kS'$, and
\begin{align}
  \left|\left\{j: N({\bf X})+1\le j\le kS', p_j = \frac{\delta}{S'}\right\}\right| = S' - N({\bf X}).
\end{align}

Hence, the Bayes estimator $\hat{P}^B({\bf X})=(a_1,a_2,\cdots,a_S)$ should minimize the posterior $\ell_1$ risk given $\bf X$ expressed as
\begin{align}\label{eq:minimization}
  &R_B({\bf X},H,n,\mu)\triangleq \sum_{i=1}^{N({\bf X})} \left|\frac{\delta}{S'}-a_i\right| + |1-\delta-a_S| \nonumber\\
  & \quad + \sum_{j=N({\bf X})+1}^{kS'} \left[\frac{(k-1)S'\left|a_j\right|}{kS'-N({\bf X})}+\frac{S'-N({\bf X})}{kS'-N({\bf X})}\left|\frac{\delta}{S'}-a_j\right|\right] ,
\end{align}
and the solution is $a_1=\cdots=a_{N({\bf X})}=\frac{\delta}{S'}, a_{N({\bf X})}=\cdots=a_{kS'}=0$ and $a_{S}=1-\delta$, and
\begin{align}\label{eq:bayes_X_1}
 R_B({\bf X},H,n,\mu) = \left(1-\frac{N({\bf X})}{S'}\right)\delta.
\end{align}

In light of this, the Bayes risk can be expressed as $R_B(H,n,\mu)=\bE[R_B({\bf X},H,n,\mu)]$, where the expectation is taken with respect to $\bf X$, and by (\ref{eq:bayes_X_1}) we only need to compute $\bE N({\bf X})$. Due to the symmetry of the Multinomial distribution, we can assume without loss of generality that ${\bf X}\sim \mathsf{Multi}(n;\frac{\delta}{S'},\cdots,\frac{\delta}{S'},0,\cdots,0,1-\delta)$, then
\begin{equation}
\begin{split}
  \bE N({\bf X}) &= \bE\left[\sum_{i=1}^{S'} \mathbbm{1}(X_i>0)\right] = \sum_{i=1}^{S'} \mathbb{P}(X_i>0)\\
  & = S'\cdot \left[1-\left(1-\frac{\delta}{S'}\right)^{n}\right] \le n\delta
\end{split}
\end{equation}
which yields
\begin{equation}
\begin{split}
  R_B(H,n,\mu)&=\bE[R_B({\bf X},H,n,\mu)] \\
  &= \left(1-\frac{\bE N({\bf X})}{S'}\right)\delta = \left(1-\frac{n\delta}{S'}\right)\delta.
\end{split}
\end{equation}

Fix $c\in(0,1)$, we set $\delta = \frac{cH}{\ln n} \le c$, and
\begin{align}
  S' &= \delta\exp\left(\frac{H}{\delta}\right)\left((1-\delta)^{\frac{1}{\delta}}\right)^{1-\delta} \\
  & \ge \frac{cH}{\ln n}\cdot n^{\frac{1}{c}}\cdot (1-c)^{\frac{1}{c}}.
\end{align}
Since the minimax risk is lower bounded by any Bayes risk, we have
\begin{align}
  \inf_{\hat{P}}\sup_{P:H(P)\le H} \bE_P\|\hat{P}-P\|_1 &\ge  R_B(H,n,\mu)\\
  &\ge \frac{cH}{\ln n}\cdot \left(1-n^{1-\frac{1}{c}}(1-c)^{-\frac{1}{c}}\right),
\end{align}
which completes the proof of the lower bound in Theorem \ref{th_lower_H}.

\setcounter{equation}{92}
For the lower bound in Theorem \ref{th_lower_H_mle}, note that we need to constrain that the Bayes estimator $\hat{P}^B$ form a probability mass, i.e., $\sum_{i=1}^S a_i=1$ in the minimization process of (\ref{eq:minimization}). Defining $\lambda\triangleq \frac{(k-2)S'+N({\bf X})}{kS'-N({\bf X})} \in [0,1]$, the derivations on the bottom of this page show that the solution becomes $a_1=\cdots=a_{N({\bf X})}=\frac{\delta}{S'}, a_{N({\bf X})}=\cdots=a_{kS'}=\frac{\delta(S'-N({\bf X}))}{S'(kS'-N({\bf X}))}$ and $a_{S}=1-\delta$. Hence, the corresponding Bayes risk given $\bf X$ is
\begin{align}\label{eq:bayes_X_2}
 R_B({\bf X},H,n,\mu) &= \frac{2(k-1)S'}{kS'-N({\bf X})}\left(1-\frac{N({\bf X})}{S'}\right)\delta \\
 & \ge \frac{2(k-1)\delta}{k}\left(1-\frac{N({\bf X})}{S'}\right).
\end{align}

Applying the similar steps, we can show that the overall Bayes risk is
\begin{align}
  R_B(H,n,\mu) &= \bE R_B({\bf X},H,n,\mu) \\
  &\ge \frac{2(k-1)cH}{k\ln n}\cdot \left(1-n^{1-\frac{1}{c}}(1-c)^{-\frac{1}{c}}\right).
\end{align}
Since the minimax risk is lower bounded by any Bayes risk, we have
\begin{align}
  &\inf_{\hat{P}\in\mathcal{M}}\sup_{P:H(P)\le H} \bE_P\|\hat{P}-P\|_1 \ge  R_B(H,n,\mu)\\
   &\qquad \qquad \ge \frac{2(k-1)cH}{k\ln n}\cdot\left(1-n^{1-\frac{1}{c}}(1-c)^{-\frac{1}{c}}\right),
\end{align}
which completes the proof of the lower bound in Theorem \ref{th_lower_H_mle} by letting $k\to \infty$.

\section{Acknowledgement}
We would like to thank I. Sason, the associate editor, and three anonymous reviewers for their very helpful suggestions. J. Jiao is partially supported by a Stanford Graduate Fellowship.

\appendices
\section{Auxiliary Lemmas}\label{sec_app_A}
The following lemma gives a sharp estimate of the Binomial mean absolute deviation.
\begin{lemma}\label{lem_bino}
  \cite{Berend2013sharp} For $X\sim\mathsf{B}(n,p)$, we have
  \begin{align}
    \bE \left|\frac{X}{n}-p\right| \le \min\left\{\sqrt{\frac{p(1-p)}{n}},2p\right\}.
  \end{align}
  Moreover, for $p<1/n$, there is an identity
  \begin{align}
    \bE\left|\frac{X}{n}-p\right| = 2p(1-p)^n.
  \end{align}
\end{lemma}

The following lemma gives some tail bounds for Poisson or Binomial random variables.
\begin{lemma}\label{lem_chernoff}
\cite[Exercise 4.7]{mitzenmacher2005probability} If $X\sim \mathsf{Poi}(\lambda)$ or $X\sim \mathsf{B}(n,\frac{\lambda}{n})$, then for any $0<\delta<1$, we have
\begin{align}
\bP(X \geq (1+\delta) \lambda) & \leq \left( \frac{e^\delta}{(1+\delta)^{1+\delta}} \right)^\lambda, \\
\bP(X \leq (1-\delta) \lambda) & \leq  \left( \frac{e^{-\delta}}{(1-\delta)^{1-\delta}} \right)^\lambda\leq  e^{-\delta^2\lambda/2}.
\end{align}
\end{lemma}

The following lemma presents the Hoeffding bound.
\begin{lemma}\label{lem_hoeffding}
  \cite{Hoeffding1963probability} For independent and identically distributed random variables $X_1,\cdots,X_n$ with $a\le X_i\le b$ for $1\le i\le n$, denote $S_n=\sum_{i=1}^n X_i$, we have for any $t>0$,
  \begin{align}
    \bP\left\{|S_n-\bE[S_n]|\ge t\right\} \le 2\exp\left(-\frac{2t^2}{n(b-a)^2}\right).
  \end{align}
\end{lemma}

A non-negative loss function $l(\cdot)$ on $\mathbb{R}^p$ is called \emph{bowl-shaped} iff $l(u)=l(-u)$ for all $u\in\mathbb{R}^p$ and for any $c\ge0$, the sublevel set $\{u: l(u)\le c\}$ is convex. The following theorem is one of the key theorems in the definition of asymptotic efficiency.
\begin{theorem}\label{lem_LAM}
  \cite[Thm. 8.11]{Vandervaart2000} Let the experiment $(P_\theta,\theta\in\Theta)$ be differentiable in quadratic mean at $\theta$ with nonsingular Fisher information matrix $I_\theta$. Let $\psi(\cdot)$ be differentiable at $\theta$. Let $\{T_n\}$ be any estimator sequence in the experiments $(P_\theta^n,\theta\in\mathbb{R}^k)$. Then for any bowl-shaped loss function $l$,
\begin{align}
&\sup_I\liminf_{n\to\infty}\sup_{h\in I}\bE_{\theta+\frac{h}{\sqrt{n}}}l\left(\sqrt{n}\left(T_n-\psi
    \left(\theta+\frac{h}{\sqrt{n}}\right)\right)\right)\nonumber\\
    &\qquad\qquad \ge \bE l(X),
\end{align}
where $\mathcal{L}(X)=\mathcal{N}(0,\psi'(\theta)I_\theta^{-1}\psi'(\theta)^T)$, and the first supremum is taken over all finite subsets $I\subset \mathbb{R}^k$.
\end{theorem}

The next lemma relates the minimax risk under the Poissonized model of an approximate probability distribution and that under the Multinomial model of a true probability distribution, where the set of approximate probability distribution is defined by
\begin{align}
  \mathcal{M}_S(\epsilon) \triangleq \left\{P=(p_1,p_2,\cdots,p_S): p_i\ge 0, \left|\sum_{i=1}^S p_i-1\right|<\epsilon \right\}.
\end{align}

We define the minimax risk for Multinomial model with $n$ observations on support size $S$ for estimating $P$ as
\begin{align}
  R(S,n) = \inf_{\hat{P}}\sup_{P\in\mathcal{M}_S}\bE_{\textrm{Multinomial}}\|\hat{P}-P\|_1,
\end{align}
and the corresponding minimax risk for Poissonized model for estimating an approximate distribution as
\begin{align}
  R_P(S,n,\epsilon) = \inf_{\hat{P}}\sup_{P\in\mathcal{M}_S(\epsilon)}\bE_{\textrm{Poissonized}}\|\hat{P}-P\|_1.
\end{align}

\begin{lemma}\label{lem_relate}
  The minimax risks under the Poissonized model and the Multinomial model are related via the following inequality: for any $\zeta\in(0,1]$ and $0<\epsilon\le\frac{\zeta}{2(1+\zeta)}$, we have
  \begin{align}
    R(S,n) \ge R_P(S,(1+\zeta)n,\epsilon) - \exp(-\frac{\zeta^2n}{24}) - \epsilon.
  \end{align}
\end{lemma}

The following lemma establishes the relationship of the $R_P(S,n,\epsilon)$ and the Bayes risk under some prior $\mu$.
\begin{lemma}\label{lem_bayes}
  Assigning prior $\mu$ to a non-negative vector $P$, denote the corresponding Bayes risk for estimating $P$ under $\ell_1$ loss by $R_B(S,n,\mu)$. If there exists a constant $A>0$ such that
  \begin{align}
    \mu\left\{P: \sum_{i=1}^Sp_i \le A\right\} = 1,
  \end{align}
  then the following inequality holds:
  \begin{align}
    R_P(S,n,\epsilon) \ge R_B(S,n,\mu) - 3A\cdot\mu\left(\mathcal{M}_S(\epsilon)^c\right).
  \end{align}
\end{lemma}

\section{Proof of Corollaries and Auxiliary Lemmas}\label{sec_app_B}
\subsection{Proof of Corollary \ref{cor_n_S}}
Consider the discrete distribution $P=(p_1,p_2,\cdots,p_S)$ with cardinality $S$, and we take $\theta=(p_1,p_2,\cdots,p_{S-1})$ to be the free parameter. By definition, we know that the Fisher information matrix is
\begin{align}
  I_{i,j}(\theta) = \bE_\theta\left[-\frac{\partial^2}{\partial \theta_i\partial \theta_j}\ln p(x|\theta)\right], \quad 1\le i,j\le S-1.
\end{align}

It is straightforward to obtain that
\begin{align}
  I_{i,j}(\theta) &= p_i\left[-\frac{\partial^2}{\partial p_i^2}\ln p_i\right]\delta_{i,j} \nonumber\\
  & \qquad + \left(1-\sum_{k=1}^{S-1}p_k\right)
\left[-\frac{\partial^2}{\partial p_i\partial p_j}\ln\left(1-\sum_{k=1}^{S-1}p_k\right)\right] \\
& = \frac{\delta_{i,j}}{p_i} + \frac{1}{p_S}.
\end{align}
where $\delta_{i,j}$ equals one if $i=j$ and zero otherwise. Hence, in matrix form we have
\begin{align}
  {\bf I}(\theta) = \bm{\Lambda} + \frac{1}{p_S}{\bf 1}{\bf 1}^{\bf T},
\end{align}
where $\bm{\Lambda}\triangleq \text{diag}(p_1^{-1},\cdots,p_{S-1}^{-1})$, and ${\bf 1}\triangleq (1,1,\cdots,1)^{\bf T}$ is a $(S-1)\times1$ column vector. According to the Woodbury matrix identity
\begin{align}
  ({\bf A} + {\bf UCV})^{-1} = {\bf A}^{-1} - {\bf A}^{-1}{\bf U}({\bf C}^{-1}+{\bf VA}^{-1}{\bf U})^{-1}{\bf V}{\bf A}^{-1},
\end{align}
we can take ${\bf A}=\bm{\Lambda}, {\bf U}={\bf 1}, {\bf C}=p_S^{-1}, {\bf V}={\bf 1}^{\bf T}$ to obtain
\begin{align}
  {\bf I}(\theta)^{-1} &= (\bm{\Lambda} + \frac{1}{p_S}{\bf 1}{\bf 1}^{\bf T})^{-1}\\
& = \bm{\Lambda}^{-1} - \bm{\Lambda}^{-1}{\bf 1}(p_S+{\bf 1}\bm{\Lambda}^{-1}{\bf 1}^{\bf T})^{-1}{\bf 1}^{\bf T}\bm{\Lambda}^{-1}\\
& = \bm{\Lambda}^{-1} - \bm{\Lambda}^{-1}{\bf 1}{\bf 1}^{\bf T}\bm{\Lambda}^{-1}.
\end{align}

After some algebra we can show that
\begin{align}
  [{\bf I}(\theta)^{-1}]_{i,j} = -p_ip_j + p_i\delta_{i,j},
\end{align}
then by choosing $l: \mathbb{R}^S\mapsto\mathbb{R}_+$ defined by $l({\bf X})\triangleq\sum_{i=1}^S|X_i|$ and $\psi((p_1,p_2,\cdots,p_{S-1}))=(p_1,p_2,\cdots,p_{S-1},1-\sum_{i=1}^{S-1}p_i)$ in Theorem \ref{lem_LAM}, for $\mathcal{L}({\bf X})=\mathcal{N}(0,\psi'(\theta){\bf I}(\theta)^{-1}\psi'(\theta)^T)$,
\begin{align}
  \mathbb{E}l({\bf X}) = \sqrt{\frac{2}{\pi}}\sum_{i=1}^{S} \sqrt{p_i(1-p_i)}.
\end{align}

If we choose $\theta=(1/S,1/S,\cdots,1/S)$, then for any estimator sequence $\{T_n\}_{n=1}^\infty$, Theorem \ref{lem_LAM} yields
\begin{align}
&\sup_I\liminf_{n\to\infty}\sup_{h\in I}\bE_{\theta+\frac{h}{\sqrt{n}}}\sqrt{n}\left\|T_n-\left(\theta+\frac{h}{\sqrt{n}}\right)\right\|_1\nonumber\\
&\qquad \ge\sqrt{\frac{2}{\pi}}\sum_{i=1}^{S} \sqrt{p_i(1-p_i)} = \sqrt{\frac{2(S-1)}{\pi}}
\end{align}
and the proof is completed by noticing that
\begin{align}
  &\liminf_{n\to\infty}\sqrt{n}\cdot \inf_{T_n}\sup_{P\in\mathcal{M}_S}\bE_P\left\|T_n-P\right\|_1 \nonumber\\
  &\qquad \ge \sup_I\liminf_{n\to\infty}\sup_{h\in I}\bE_{\theta+\frac{h}{\sqrt{n}}}\sqrt{n}\left\|T_n-\left(\theta+\frac{h}{\sqrt{n}}\right)\right\|_1.
\end{align}

\subsection{Proof of Lemma \ref{lem_relate}}
By the definition of the minimax risk under the Multinomial model, for any $\delta>0$, there exists an estimator $\hat{P}_M({\bf X},S,n)$ such that
\begin{align}
  \sup_{P\in\mathcal{M}_S} \mathbb{E}_P\|\hat{P}_M({\bf X},S,n)-P\|_1 < R(S,n) + \delta, \quad \forall n.
\end{align}
Now we construct a new estimator under the Poissonized model, i.e., we set $\hat{P}_P({\bf X},S)\triangleq \hat{P}_M({\bf X},S,n')$ where $n'=\sum_{i=1}^S X_i\sim\mathsf{Poi}(n\sum_{i=1}^Sp_i)$. Then we can obtain that for $0<\epsilon<\frac{\zeta}{2(1+\zeta)}$ and $\zeta\in(0,1)$,
\begin{align}
  &R_P(S,n,\epsilon) \le \sup_{P\in\mathcal{M}_S(\epsilon)} \mathbb{E}_P\|\hat{P}_P({\bf X},S)-P\|_1\\
  &\qquad = \sum_{m=0}^\infty\sup_{P\in\mathcal{M}_S(\epsilon)} \mathbb{E}_P\|\hat{P}_M({\bf X},S,m)-P\|_1\cdot\mathbb{P}(n'=m)\\
  &\qquad\le \sum_{m=0}^\infty\sup_{P\in\mathcal{M}_S(\epsilon)} \left(\mathbb{E}_P\left\|\hat{P}_M({\bf X},S,m)-\frac{P}{\sum_{i=1}^Sp_i}\right\|_1 \right.\nonumber\\
  &\qquad\qquad \left. + \left\|\frac{P}{\sum_{i=1}^Sp_i}-P\right\|_1\right)\cdot\mathbb{P}(n'=m)\\
  &\qquad\le \sum_{m=0}^\infty \left(R(S,m)+\delta+\epsilon\right)\cdot\mathbb{P}(n'=m)\\
  &\qquad\le \mathbb{P}(n'\le \frac{n}{1+\zeta}) + R(S,\frac{n}{1+\zeta}) + \epsilon + \delta\\\label{eq:relate}
  &\qquad\le \exp(-\frac{\zeta^2n}{24}) + R(S,\frac{n}{1+\zeta}) + \epsilon + \delta,
\end{align}
where we have used Lemma \ref{lem_chernoff} in the last step. The desired result follows directly from (\ref{eq:relate}) and the arbitrariness of $\delta$.

\subsection{Proof of Lemma \ref{lem_bayes}}
Denote the conditional prior $\pi$ by
\begin{align}
  \pi(A) = \frac{\mu(A\cap \mathcal{M}_S(\epsilon))}{\mu(\mathcal{M}_S(\epsilon))},
\end{align}
we consider the Bayes estimator $\hat{P}'$ under prior $\pi$ and the corresponding Bayes risk $R_B'(S,n,\pi)$. Due to our construction of the prior, we know that for all $\bf X$, the sum of all entries of $\hat{P}'({\bf X})$ will not exceed $A$ almost surely. Since $R_B(S,n,\mu)$ is the Bayes risk under prior $\mu$, applying $\hat{P}'$ yields
\begin{align}
  R_B(S,n,\mu) &\le \int \mathbb{E}_P\|\hat{P}'-P\|_1 \mu(dP)\\
  &= \int_{\mathcal{M}_S(\epsilon)}\mathbb{E}_P\|\hat{P}'-P\|_1 \mu(dP) \nonumber\\
  &\qquad + \int_{\mathcal{M}_S(\epsilon)^c} \mathbb{E}_P\|\hat{P}'-P\|_1 \mu(dP)\\
  &\le \frac{1}{\mu(\mathcal{M}_S(\epsilon))} \int_{\mathcal{M}_S(\epsilon)} \mathbb{E}_P\|\hat{P}'-P\|_1 \pi(dP)\nonumber\\
  &\qquad+ \int_{\mathcal{M}_S(\epsilon)^c}2A \mu(dP)\\ \label{eq:inequality_1}
  &= \frac{R_B'(S,n,\pi)}{\mu(\mathcal{M}_S(\epsilon))} + 2A\left(1-\mu(\mathcal{M}_S(\epsilon))\right).
\end{align}
Since the Bayes risk serves as a lower bound for the minimax risk, i.e., $R_P(S,n,\epsilon)\ge R_B'(S,n,\pi)$, we have
  \begin{align}
    R_P(S,n,\epsilon) \ge R_B(S,n,\mu) - (2A+R_B(S,n,\mu))\mu\left(\mathcal{M}_S(\epsilon)^c\right).
  \end{align}
Then the proof is completed by noticing that $R_B(S,n,\mu)\le A$, for the risk of the null estimator $\hat{P}(\bf X)={\bf 0}$ under prior $\mu$ does not exceed $A$.

\bibliographystyle{IEEEtran}
\bibliography{reference}

\begin{IEEEbiographynophoto}{Yanjun Han}
(S'14) received his B.Eng. degree with the highest honor in Electronic Engineering from Tsinghua University, Beijing, China in 2015. He is currently working towards the Ph.D. degree in the Department of Electrical Engineering at Stanford University. His research interests include information theory and statistics, with applications in communications, data compression, and learning.
\end{IEEEbiographynophoto}

\begin{IEEEbiographynophoto}{Jiantao Jiao}
(S'13) received his B.Eng. degree with the highest honor in Electronic Engineering from Tsinghua University, Beijing, China in 2012, and a Master's degree in Electrical Engineering from Stanford University in 2014. He is currently working towards the Ph.D. degree in the Department of Electrical Engineering at Stanford University. He is a recipient of the Stanford Graduate Fellowship (SGF). His research interests include information theory and statistical signal processing, with applications in communication, control,
computation, networking, data compression, and learning.
\end{IEEEbiographynophoto}

\begin{IEEEbiographynophoto}{Tsachy Weissman}
(S'99-M'02-SM'07-F'13) graduated summa cum laude with a B.Sc. in electrical engineering from the Technion in 1997, and earned his Ph.D. at the same place in 2001. He then worked at Hewlett Packard Laboratories with the information theory group until 2003, when he joined Stanford University, where he is currently Professor of Electrical Engineering and incumbent of the STMicroelectronics chair in the School of Engineering. He has spent leaves at the Technion, and at ETH Zurich.

Tsachy's research is focused on information theory, compression, communication, statistical signal processing, the interplay between them, and their applications. He is recipient of several best paper awards, and prizes for excellence in research and teaching. He served on the editorial board of the \textsc{IEEE Transactions on Information Theory} from Sept. 2010 to Aug. 2013, and currently serves on the editorial board of Foundations and Trends in Communications and Information Theory. He is Founding Director of the Stanford Compression Forum.
\end{IEEEbiographynophoto}
\end{document}